%% file: asilomar2022-mainfold.tex
\documentclass[conference, 10pt]{IEEEtran}
 \IEEEoverridecommandlockouts
 \usepackage{etex}
 \reserveinserts{100}
\usepackage[utf8]{inputenc}
\usepackage[pdftex]{graphicx}
\usepackage{graphicx}
\usepackage[T1]{fontenc}    
\usepackage{hyperref}       
\usepackage{url}            
\usepackage{booktabs}       
\usepackage{amsfonts}       
\usepackage{nicefrac}       
\usepackage{microtype}      
\usepackage{graphics, theorem, times, cite}
\usepackage{epstopdf}
\usepackage{tikz}
\usepackage{setspace}

\usepackage{bbm}
\usepackage{mathrsfs}
\usepackage{adjustbox}
\usepackage{makecell}
\usepackage{wrapfig}

\usepackage{amssymb,amsmath,amsfonts,bm,epsfig,theorem,epstopdf}
\usepackage{rotating,setspace,latexsym,epsf,color,dsfont,caption,mathtools}
\usepackage{subfigure}

\usepackage{algorithm}
\usepackage{algorithmic}
\usepackage[shortlabels]{enumitem}

\input{mysymbol.sty}
\input{my_sections}

\newcommand{\QED}{\hfill\ensuremath{\blacksquare}}

\def\E{\E}

\newtheorem{assumption}{\hspace{0pt}\bf Assumption}
\newenvironment{proof}[1][Proof]{{\it #1. } }{\ \rule{0.5em}{0.5em}}

\newtheorem{proposition}{\hspace{0pt}\bf Proposition}

\newtheorem{theorem}{\hspace{0pt}\bf Theorem}

\newtheorem{definition}{\hspace{0pt}\bf Definition}

\title{Convolutional Neural Networks on Manifolds: \\From Graphs and Back}
\author{Zhiyang Wang \quad Luana Ruiz \quad Alejandro Ribeiro \thanks{Supported by NSF CCF 1717120, Theorinet Simons and ARL DCIST CRA under Grant W911NF-17-2-0181. Zhiyang and Alejandro are with Department of Electrical and Systems Engineering, University of Pennsylvania, Philadelphia, Pennsylvania, USA. Luana is with Simons-Berkeley Institute, California, USA.}
}

\begin{document}
\maketitle
\begin{abstract}
Geometric deep learning has gained much attention in recent years due to more available data acquired from non-Euclidean domains. Some examples include point clouds for 3D models and wireless sensor networks in communications. Graphs are common models to connect these discrete data points and capture the underlying geometric structure. With the large amount of these geometric data, graphs with arbitrarily large size tend to converge to a limit model -- the manifold. Deep neural network architectures have been proved as a powerful technique to solve problems based on these data residing on the manifold. In this paper, we propose a manifold neural network (MNN) composed of a bank of manifold convolutional filters and point-wise nonlinearities. We define a manifold convolution operation which is consistent with the discrete graph convolution by discretizing in both space and time domains. To sum up, we focus on the manifold model as the limit of large graphs and construct MNNs, while we can still bring back graph neural networks by the discretization of MNNs.  We carry out experiments based on point-cloud dataset to showcase the performance of our proposed MNNs. 
\end{abstract}

\begin{IEEEkeywords}
Manifold convolution, manifold neural networks, geometric deep learning 
\end{IEEEkeywords}

\section{Introduction}
\label{sec:label}
\input{introduction}

\section{Manifold Convolution}
\label{sec:manifold_conv}
\input{manifoldconvolution}

\section{Manifold Neural Networks}
\label{sec:mnn}
\input{mnn}

\section{Discretization in Space and Time}
\label{sec:discret}
\input{discretization}

\section{Simulations}
\label{sec:sim}
\input{simulations}

\section{Conclusion}
\label{sec:conclusion}
\input{conclusion}

\urlstyle{same}
\bibliographystyle{IEEEtran}
\bibliography{references}

\appendix
 {\section{Appendix}
 \input{appendix}}
 
\end{document}

%% file: my_sections.tex
\usepackage{needspace}



%% file: introduction.tex
Convolutional neural networks (CNNs) have achieved impressive success in a wide range of applications, including but not limited to natural language processing \cite{wang2018application}, image denoising \cite{zhang2018ffdnet} and video analysis \cite{xu2019spatiotemporal}. Convolution operations are implemented to capture the local information and features based on the characteristics of the dataset. The remarkable success provides the support that CNNs are recognized as powerful techniques when processing traditional signals such as sound, image or video, which all lie in the Euclidean domains. As we have more access to larger scale data and stronger computing power, increasing attention is being paid to processing data lying in the non-Euclidean domains.

Many practical problems rely on non-Euclidean data. There is the case, for example, detection and recommendation in social networks \cite{aggarwal2020machine}, resource allocations over wireless networks \cite{wang2022learning}, point clouds for shape segmentation \cite{xie2020linking}. There have been works that extend the CNN architecture to non-Euclidean domains \cite{gama2019convolutional, defferrard2016convolutional, wang2021stability}, which reproduce the success of CNNs in Euclidean domains. Among these models, graphs are commonly used to construct the underlying data structure, while the graph size scales with the amount of data. In this work, we aim to construct CNNs on this more general model -- the manifold.

Graphs with well-defined limits are shown to converge to a manifold model \cite{belkin2008towards, bronstein2017geometric}, which makes the manifold capable of capturing properties for a series of graphs. The convolution operation is not taken for granted in non-Euclidean domains due to the lack of global parametrization and shift invariance. We define a manifold convolution operation based on the heat diffusion process controlled by the Laplace-Beltrami operator. We construct a manifold convolutional filter to process manifold signals. By cascading the layers consisting of manifold filter banks and nonlinearities, we can define the manifold neural networks (MNNs) as a deep learning framework on the manifold. To motivate the practical implementations of our proposed MNNs, we first discretize the MNN in the space domain by sampling points on the manifold. The proposed MNN can be transferred to this discretized manifold as a discretized MNN which converges to the underlying MNN when the manifold signal is bandlimited. We further carry out discretization in the time domain by sampling the filter impluse function in discrete and finite time steps. In this way, we can not only execute our proposed MNNs, but also recover the graph convolutions and graph neural networks \cite{gama2019convolutional}. This concludes our
thought starting from a graph sequence to the limit as a a manifold and back to the graphs. We finally verify the performance of our proposed MNN with a point cloud based model classification problem.

Related works include neural networks built on graphons \cite{ruiz2020graphon, ruiz2021graphonsignal}, which are limits of a sequence of dense graphs. Different from manifolds, graphons only represent the limits for graphs with unbounded degrees \cite{lovasz2012large}. Stability of MNNs have been studied considering the perturbations to the Laplace-Beltrami operator \cite{wang2021stability, wang2021stabilityrela}. A general framework for algebraic neural networks has been proposed for architectures unified with commutative algebras \cite{parada2020algebraic}.

The rest of the paper is organized as follows. We start with some preliminary concepts and define the manifold convolutions in Section \ref{sec:manifold_conv}. We construct the MNNs based on manifold filters in Section \ref{sec:mnn}. In Section \ref{sec:discret}, we implement the discretization in space and time domains to make the MNNs realizable which also bring back to graph convolutions. Our proposed MNN is verified in a model classification problem in Section \ref{sec:sim}. The conclusions are presented in Section \ref{sec:conclusion}.

%% file: manifoldconvolution.tex
\subsection{Preliminary Definitions}
In this paper, we consider a compact, smooth, and differentiable $d$-dimensional submanifold $\ccalM$ embedded in $\reals^N$. The embedding induces a Riemannian structure \cite{gallier2020differential} on $\ccalM$ which endows a measure $\mu$ over the manifold. \emph{Manifold signals} supported on $\ccalM$ are smooth scalar functions $f: \ccalM \rightarrow \reals$. We consider manifold signals in a Hilbert space in which we define the inner product as
\begin{equation}\label{eqn:innerproduct}
    \langle f,g \rangle_{L^2(\ccalM)}=\int_\ccalM f(x)g(x) \text{d}\mu(x) 
\end{equation}
with the norm defined as $\|f\|^2_{L^2(\ccalM)}=\langle f, f\rangle_{L^2(\ccalM)}$. 

The manifold is locally Euclidean, which elicits \emph{intrinsic gradient} for differentiation as a local operator \cite{bronstein2017geometric}. The local Euclidean space around $x\in\ccalM$ containing all of the vectors tangent to $\ccalM$ at $x$ is denoted as tangent space $T_x\ccalM$. We use $T\ccalM$ to represent the disjoint union of all tangent spaces on $\ccalM$. The intrinsic gradient can thus be written as an operator $\nabla: L^2(\ccalM)\rightarrow L^2(T\ccalM)$ mapping scalar functions to tangent vector functions on $\ccalM$. The adjoint operator of intrinsic gradient is the \emph{intrinsic divergence} defined as $\text{div}: L^2(T\ccalM)\rightarrow L^2(\ccalM)$. Based on these two differentiation operators, the Laplace-Beltrami (LB) operator $\ccalL: L^2(\ccalM) \to L^2(\ccalM)$ can be defined as the intrinsic divergence of the intrinsic gradient \cite{rosenberg1997laplacian}, formally as
\begin{equation}\label{eqn:Laplacian}
    \ccalL f=-\text{div}\circ \nabla f=-\nabla \cdot \nabla f.
\end{equation}
Similar to the Laplacian operator in Euclidean domains or the Laplace matrix in graphs \cite{moon2012field}, the LB operator evaluates how much the function value at point $x$ differs from the average function value of its neighborhood \cite{bronstein2017geometric}.  

The LB operator provides a basis for expressing and solving physical tasks by Partial Differential Equations (PDEs). One of the remarkable applications is characterizing the heat diffusion over manifolds by the \emph{heat equation}
\begin{equation}\label{eqn:heat}
    \frac{\partial u(x,t)}{\partial t}+\ccalL u(x,t)=0 \text{,}
\end{equation}
where $u(x,t)\in L^2(\ccalM)$ measures the temperature at $x\in \ccalM$ at time $t\in\reals^+$. With initial condition given by $u(x,0)=f(x)$, the solution can be expressed as
\begin{equation}\label{eqn:heat-solution}
    u(x,t) = e^{-t \ccalL}f(x) \text{,}
\end{equation}
which provides an essential element to construct manifold convolution in the following.

Due to the compactness of $\ccalM$, the LB operator $\ccalL$ is self-adjoint and positive-semidefinite. This means that $\ccalL$ possesses a real positive spectrum $\{\lambda_i\}_{i=1}^\infty$ with the eigenvalues $\lambda_i$ and the corresponding eigenfunctions $\phi_i$ satisfying 
\begin{equation}\label{eqn:laplacian-decomp}
\ccalL \bm\phi_i =\lambda_i \bm\phi_i
\end{equation}

The eigenvalues are ordered as $0<\lambda_1\leq \lambda_2\leq \hdots$. According to Weyl's law \cite{arendt2009weyl}, we have $\lambda_i \propto i^{2/d}$ for a $d$-dimensional manifold. The orthonormal eigenfunctions $\phi_i$ form a general eigenbasis of $L^2(\ccalM)$ in the intrinsic sense. 

Since $\ccalL$ is a total variation operator, the eigenvalues $\lambda_i$ can be interpreted as the canonical frequencies and the eigenfunctions $\bm\phi_i$ as the canonical oscillation modes of $\ccalM$.


\subsection{Manifold Convolutional Filters}
Convolution operation is a powerful technique to give zero state response to any input signal with the filter impulse response of the system \cite{oppenheim1997signals}. Similar to time signals are processed with time convolutions and graph signals are processed by graph convolutions \cite{gama2020graphs}, we define manifold convolution with a filter impulse response $\tdh$ and manifold signal $f$. 

\begin{definition}[Manifold filter]
\label{def:manifold-convolution}
Let $\tdh:\reals^+ \to \reals$ and let $f \in L^2(\ccalM)$ be a manifold signal. The manifold filter with impulse response $\tdh$, denoted $\bbh$, is given by
\begin{align} \label{eqn:convolution-conti}
   g(x) = (\bbh f)(x) := \int_0^\infty \tdh(t)u(x,t)\text{d}t 
\end{align}
where $u(x,t)$ is the solution of the heat equation \eqref{eqn:heat} with $u(x,0)=f(x)$. Substitute the solution $u(x,t)$ with \eqref{eqn:heat-solution}, and we can derive a parametric form of $\bbh$ as 
\begin{equation} \label{eqn:manifold-conv-spatial}
   g(x) = (\bbh f)(x) =\int_0^\infty \tdh(t)e^{-t\ccalL}f(x)\text{d}t =  \bbh(\ccalL)f(x) \text{.}
\end{equation}
\end{definition}

Manifold filters are local spatial operators operating directly on points on the manifold based on the LB operator. The exponential term $e^{-t\ccalL}$ can be interpreted as a shift operator like the time delay in a Linear-Time Invariant (LTI) filter \cite{oppenheim1997signals} and the graph shift in a Linear-Shift Invariant (LSI) graph filter \cite{gama2020graphs}. In fact, manifold filters can recover graph filters by discretization, which we discuss thoroughly in Section \ref{sec:discret}.

The LB operator $\ccalL$ possesses the eigendecomposition $\{\lambda_i, \phi_i\}_{i=1}^\infty$. Eigenvalue $\lambda_i$ can be interpreted as the canonical frequency and the eigenfunction $\phi_i$ as the canonical oscillation mode. By projecting a manifold signal $f$ onto the eigenfunction, we can write the \emph{frequency representation} $\hat{f}$ as
\begin{equation}\label{eqn:f-decomp}
[\hat{f}]_i= \langle f, \bm\phi_i \rangle_{L^2(\ccalM)} = \int_\ccalM f(x) \bm\phi_i(x) \text{d} \mu(x) \text{.}
\end{equation} 

\begin{definition}[Bandlimited manifold signals]
\label{def:bandlimited-manifold}
A manifold signal is defined as $\lambda_M$-bandlimited with $\lambda_M>0$ if $[\hat{f}]_i=0$ for all $i$ such that $\lambda_i> \lambda_M$.
\end{definition}

The spectrum and eigenbasis of the LB opertor help to understand the frequency behavior of the manifold filter $\bbh(\ccalL)$. The frequency representation of manifold filter output $g$ can be similarly written as
\begin{equation}
    [\hat{g}]_i = \int_\ccalM \int_0^\infty \tdh(t) e^{-t\ccalL} f(x) \text{d} t \bm\phi_i(x) \text{d} \mu(x) \text{.}
\end{equation}
By substituting $e^{-t\ccalL}\phi_i = e^{-t\lambda_i}\phi_i$, we can get
\begin{equation}\label{eqn:projection}
    [\hat{g}]_i = \int_0^\infty \tdh(t) e^{-t\lambda_i} \text{d}t  [\hat{f}]_i \text{.}
\end{equation}
The function solely dependent on $\lambda_i$ is defined as the \emph{frequency response} of the filter $\bbh(\ccalL)$.

\begin{definition}[Frequency response]
\label{def:frequency-response}
The frequency response of the filter $\bbh(\ccalL)$ is given by
\begin{equation}\label{eqn:operator-frequency}
\hat{h}(\lambda)=\int_0^\infty \tdh(t) e^{- t \lambda  }\text{d}t \text{,}
\end{equation}
which leads \eqref{eqn:projection} to
$
[\hat{g}]_i = \hat{h}(\lambda_i)[\hat{f}]_i \text{.}$
\end{definition}

Definition \ref{def:frequency-response} indicates that the frequency response of manifold filter is point-wise in frequency domain. Combining the frequency representation of $\hat{g}$ over the whole spectrum, we can obtain the frequency representation of manifold filter $\bbh$ as
\begin{equation}
\label{eqn:filter-frequency}
g = \bbh(\ccalL)f = \sum_{i=1}^\infty \hat{h}(\lambda_i)\langle f,\phi_i \rangle_{L^2(\ccalM)} \phi_i.
\end{equation}

The well-defined manifold filters make an important building block in Manifold Neural Networks (MNNs), which we show in the following section.

%% file: mnn.tex

Manifold neural networks (MNNs) augment manifold filters with a point-wise nonlinear activation function. We extend the definition to a mapping on the manifold $\sigma: L^2(\ccalM)\rightarrow L^2(\ccalM)$ as an independent application on each point of the manifold. In a single-layer MNN, the manifold signal $f$ is passed through a manifold filter followed by a point-wise nonlinearity as
\begin{equation}\label{eqn:mnn-1layer}
f_1(x) = \sigma\Bigg( \bbh(\ccalL) f(x)\Bigg),
\end{equation}
which can be seen as a basic nonlinear processing of the input manifold signal. By stacking this procedure in layers, a multilayer MNN can be constructed which can be formally written as a function composition. The output manifold signal of a layer becomes the input signal of the next layer. Let $l=1,2,\cdots,L$ stand for the index for the layer and $\bbh_l(\ccalL)$ as the manifold filter on each layer. For a specific layer $l$, filter $\bbh_l(\ccalL)$ takes the output $f_{l-1}(x)$ as the input produces the output of layer $l$ as 
\begin{equation}\label{eqn:mnn-multilayer}
f_l(x) = \sigma\Bigg( \bbh_l(\ccalL) f_{l-1}(x)\Bigg),
\end{equation}
where $f_0(x)=f(x)$ as the given input manifold signal. After a recursive applications through $L$ layers, we can get the output of the MNN as $f_L(x)$.

When considering multiple features in each layer to increase the representation power of MNN, the manifold filters map the input $F_{l-1}$ features from layer $l-1$ to $F_l$ intermediate features in layer $l$ with a bank of manifold filters, i.e.,
\begin{equation}\label{eqn:mnn}
y_l^p(x) =  \sum_{q=1}^{F_{l-1}} \bbh_l^{pq}(\ccalL) f_{l-1}^q(x),
\end{equation}
where $\bbh_l^{pq}(\ccalL)$ is the filter mapping the $q$-th feature from layer $l-1$ to the $p$-th feature of layer $l$, for $1\leq q\leq F_{l-1}$ and $1\leq p\leq F_{l}$. The intermediate features are then processed by the nonlinearity $\sigma$ as
\begin{equation}\label{eqn:mnn}
f_l^p(x) = \sigma\Bigg(y_l^p(x) \Bigg).
\end{equation}
The output of layer $L$ is the output of MNN with $F_L$ features. To represent the MNN more precisely, we gather the impulse responses of all the manifold filters $\bbh_l^{pq}$ as a function set $\bbH$ and define the MNN as a map $\bbPhi(\bbH,\ccalL, f)$. This map is parameterized by both the filter functions $\bbH$ and the LB operator $\ccalL$.

%% file: discretization.tex
MNNs are built based on manifold convolutional filters (Definition \ref{def:manifold-convolution}) processing manifold signals over an infinite time horizon. In practice, the continuous architectures cannot be implemented directly. In this section, we discuss the practical application of MNNs \eqref{eqn:mnn} by discretization in both space and time domains. 

\subsection{{Discretization in the Space Domain}}
Realistically, the underlying manifold along with its LB operator is inaccessible directly. It is common to use sampling points to form a point cloud as an approximation of the manifold structure. 

With the knowledge of the coordinates of the sampling points, the underlying manifold structure can be approximated by a geometric graph structure which can also be seen as a discretized manifold \cite{dunson2021spectral, belkin2008towards}. The graph Laplacian is hence the approximation of the LB operator, whose convergence to the LB operator as the number of sampling points increases has been shown explicitly \cite{calder2019improved, dunson2021spectral}.

Specifically, we model the set of $n$ sampling points as $X= \{x_1, x_2,\dots, x_n\}$ which are sampled i.i.d. from measure $\mu$ of manifold $\ccalM\subset \reals^N$. A complete weighted symmetric graph $\bbG_n$ can be constructed by seeing the sampling points as the vertices of the graph and the Euclidean distance between pairs of points as the edge weights. To be more precise, the edge weight $w_{ij}$ connecting point $x_i$ and $x_j$ is given by
\begin{equation}\label{eqn:weight}
    w_{ij}=\frac{1}{n}\frac{1}{t_n(4\pi t_n)^{k/2}}\exp\left(-\frac{\|x_i-x_j\|^2}{4t_n}\right),
\end{equation}
with $\|x_i-x_j\|$ representing the Euclidean distance between $x_i$ and $x_j$. Parameter $t_n$ controls the chosen Gaussian kernel \cite{belkin2008towards}. The adjacency matrix of $\bbG_n$ is thus defined as $[\bbA_n]_{ij}=w_{ij}$ for $1 \leq i,j\leq n$ with $\bbA_n \in \reals^{n\times n}$. The correspondent graph Laplacian matrix $\bbL_n$ \cite{merris1995survey} thus can be defined as
\begin{equation}\label{eqn:Laplacian-matrix}
    \bbL_n = \mbox{diag}(\bbA_n \boldsymbol{1})-\bbA_n,
\end{equation}
which is the approximated LB operator of the discretized manifold. We define a uniform sampling operator $\bbP_n: L^2(\ccalM)\rightarrow L^2(X)$ to discretize a bandlimited manifold signal $f$ (Definition \ref{def:bandlimited-manifold}) as a graph signal $\bbx_n\in \reals^N$, denoted as
\begin{equation}
\label{eqn:sampling}
    \bbx_n = \bbP_n f\text{ with }[\bbx_n]_i = f(x_i), \quad  x_i \in X,
\end{equation}
which indicates that the $i$-th element of $\bbx_n$ is the evaluation of manifold signal $f$ at the sampling point $x_i$.

The definition of manifold filter (Definition \ref{def:manifold-convolution}) has shown that it can be parameterized by the LB operator. Likewise, we can replace the LB operator with the discrete Laplacian operator and turns $\bbh$ to a discrete manifold filter, i.e.,
\begin{equation}\label{eqn:sample_manifold_convolution}
    \bbz_n = \int_{0}^\infty \tdh(t) e^{-t\bbL_n} \text{d}t \bbx_n=\bbh(\bbL_n)\bbx_n,\; \bbx_n, \bbz_n \in\reals^n, 
\end{equation}
where $\bbz_n$ is the output discrete graph signal. By cascading the layers containing discrete manifold filters and point-wise nonlinearities, we can construct a neural network on this discreteized manifold as
\begin{equation}\label{eqn:dis-mnn}
    \bbx_{n,l}^p = \sigma\left(\sum_{q=1}^{F_{l-1}} \bbh_l^{pq}(\bbL_n) \bbx^q_{n,l-1} \right),
\end{equation}
where $\bbh_l^{pq}$ in the filter bank maps the features in the $l-1$-th layer to the $p$-th feature in the $l$-th layer. Each layer contains $F_l$ features. By gathering all the filter impulse functions as a set $\bbH$, we can represent the neural network on the discrete manifold as a map $\bm\Phi(\bbH,\bbL_n,\bbx_n)$.

As the number of sampling points goes to infinity, the discrete graph signal $\bbx_n$ converges to the manifold signal $f$ while the discrete graph Laplacian matrix $\bbL_n$ also converges to the LB operator $\ccalL$ of the underlying manifold \cite{belkin2008towards}. Combining these convergence, we conclude that the outputs of the neural network on the discrete manifold converge to the outputs of the continuous manifold neural network, which we state explicitly as the following proposition.

\begin{proposition}
\label{prop:convergence} 
Let $X=\{x_1, x_2,...x_n\}$ be a set of $n$ sampling points i.i.d. from a bandlimited manifold signal $f$ of $d$-dimensional manifold $\ccalM \subset \reals^N$, sampled by an operator $\bbP_n$ \eqref{eqn:sampling}. Let $\bbG_n$ be a discrete approximation of $\ccalM$ constructed from $X$ with weight values set as \eqref{eqn:weight} with $t_n = n^{-1/(d+2+\alpha)}$ and $\alpha>0$. Let $\bm\Phi(\bbH, \cdot, \cdot)$ be the neural network parameterized by the LB operator $\ccalL$ of the manifold $\ccalM$ \eqref{eqn:mnn} or by the discrete Laplacian operator $\bbL_n$ of the discretized manifold $\bbG_n$. Under the assumption that the frequency response of filters in $\bbH$ are Lipschitz continuous, it holds for each $i=1,2,\hdots,n$ that
\begin{equation}
    \lim_{n\rightarrow \infty} \|\bm\Phi(\bbH, \bbL_n, \bbP_n f) - \bbP_n\bm\Phi(\bbH, \ccalL,f)\|_{L^2(\bbG_n)} = 0,
\end{equation}
with the limit taken in probability.
\end{proposition}
\begin{proof}
See Appendix in \cite{wang2022convolution}.
\end{proof}

Proposition \ref{prop:convergence} claims that neural networks constructed from the discrete Laplacian $\bbL_n$ can give a good approximation of the manifold neural networks on bandlimited manifold signals.


\subsection{Discretization in the Time Domain}
Definition \ref{def:manifold-convolution} shows that the manifold filter $\bbh$ is determined by the impulse response function $\tilde{h}(t)$. Therefore, if we want to learn the MNN, we need to learn the continuous function $\tilde{h}(t)$, which is infeasible without enough prior knowledge. To make this learning practical, we discretize function $\tilde{h}(t)$ in the continuous time domain with a fixed sampling interval $T_s$. We replace the filter response function with a series of coefficients $h_k = \tilde{h}(k T_s)$, $k =0 ,1, 2\dots$. For simplicity, we settle the sampling interval as a unit $T_s=1$. The manifold convolution in the discrete time domain can be written as
\begin{equation}
\label{eqn:manifold_convolution_discrete}
    \bbh(\ccalL) f(x)= \sum_{k=0}^{\infty} h_k e^{-k\ccalL}f(x), 
\end{equation}
where $\{h_k\}_{k=0}^\infty$ are called filter coefficients or taps.

Considering the infinite number of the filter coefficients, the realization of $\bbh(\ccalL) f(x)$ is still impractical. We fix $K$ samples over the time horizon and reformulate \eqref{eqn:manifold_convolution_discrete} as
\begin{equation}
\label{eqn:manifold_convolution_discrete}
    \bbh(\ccalL) f(x)= \sum_{k=0}^{K-1} h_k e^{-k\ccalL}f(x),
\end{equation}
which corresponds to the form of a finite impulse response (FIR) filter with shift operator $e^{-\ccalL}$. Similarly, we can discretize the filter on the discretized manifold in \eqref{eqn:sample_manifold_convolution} over the time domain, which leads to a practical manifold filter on a discretized manifold and in the discrete time domain, i.e.,
\begin{equation}
\label{eqn:discrete_manifold_convolution_discrete}
  \bbz_n=  \bbh(\bbL_n) \bbx = \sum_{k=0}^{K-1} h_k e^{-k\bbL_n}\bbx_n.
\end{equation}
We can observe that this discretized manifold filter recovers the form of graph convolution \cite{gama2020graphs} with $e^{-\bbL_n}$ seen as the graph shift operator. By replacing the filter $\bbh_l^{pq}(\bbL_n)$ in \eqref{eqn:dis-mnn} with \eqref{eqn:discrete_manifold_convolution_discrete}, we further recover the graph neural network (GNN) architecture from MNNs. Till now we have completed the process of building MNNs from graphs and back. Manifolds can be seen as the limits of graphs and thus a tool for analyzing large graphs. CNNs can be constructed on manifolds to address problems in the limit sense. By further discretizing the constructed MNNs for practical implementation, we can bring the problem back to the graphs.

%% file: simulations.tex

\begin{figure*}[t]
\centering
\includegraphics[trim=130 0 80 0,clip,width=0.12\textwidth]{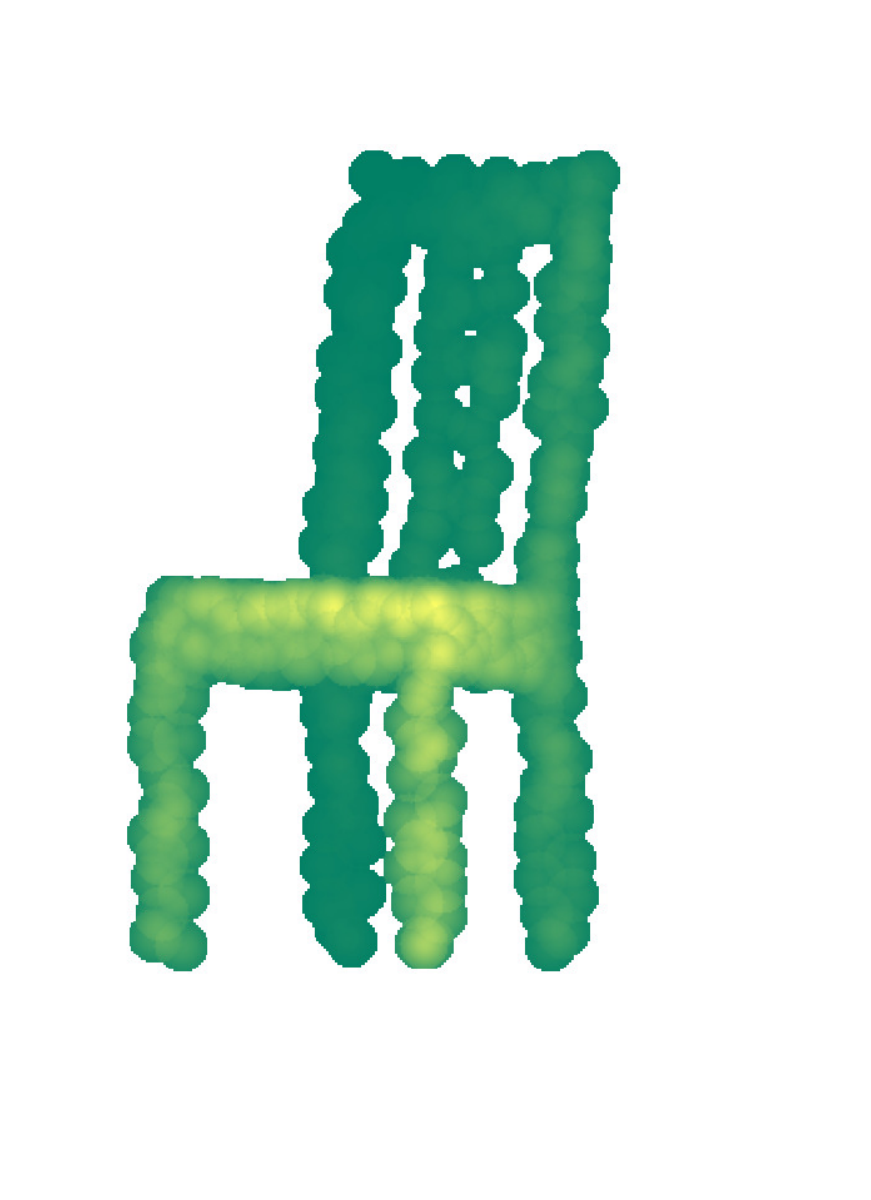}
\includegraphics[trim=90 0 100 0,width=0.12\textwidth]{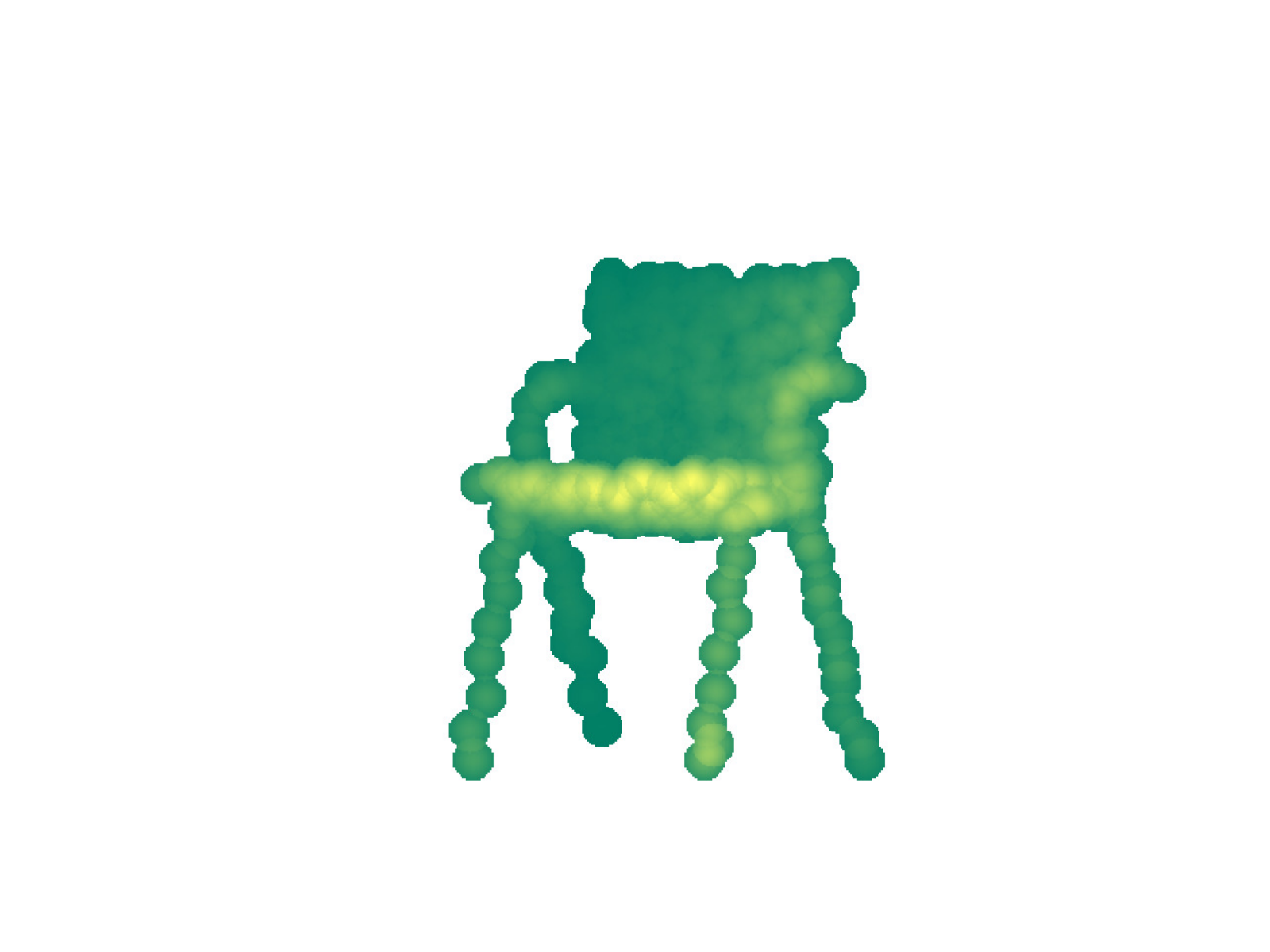}  
\includegraphics[trim=60 0 100 0,width=0.15\textwidth]{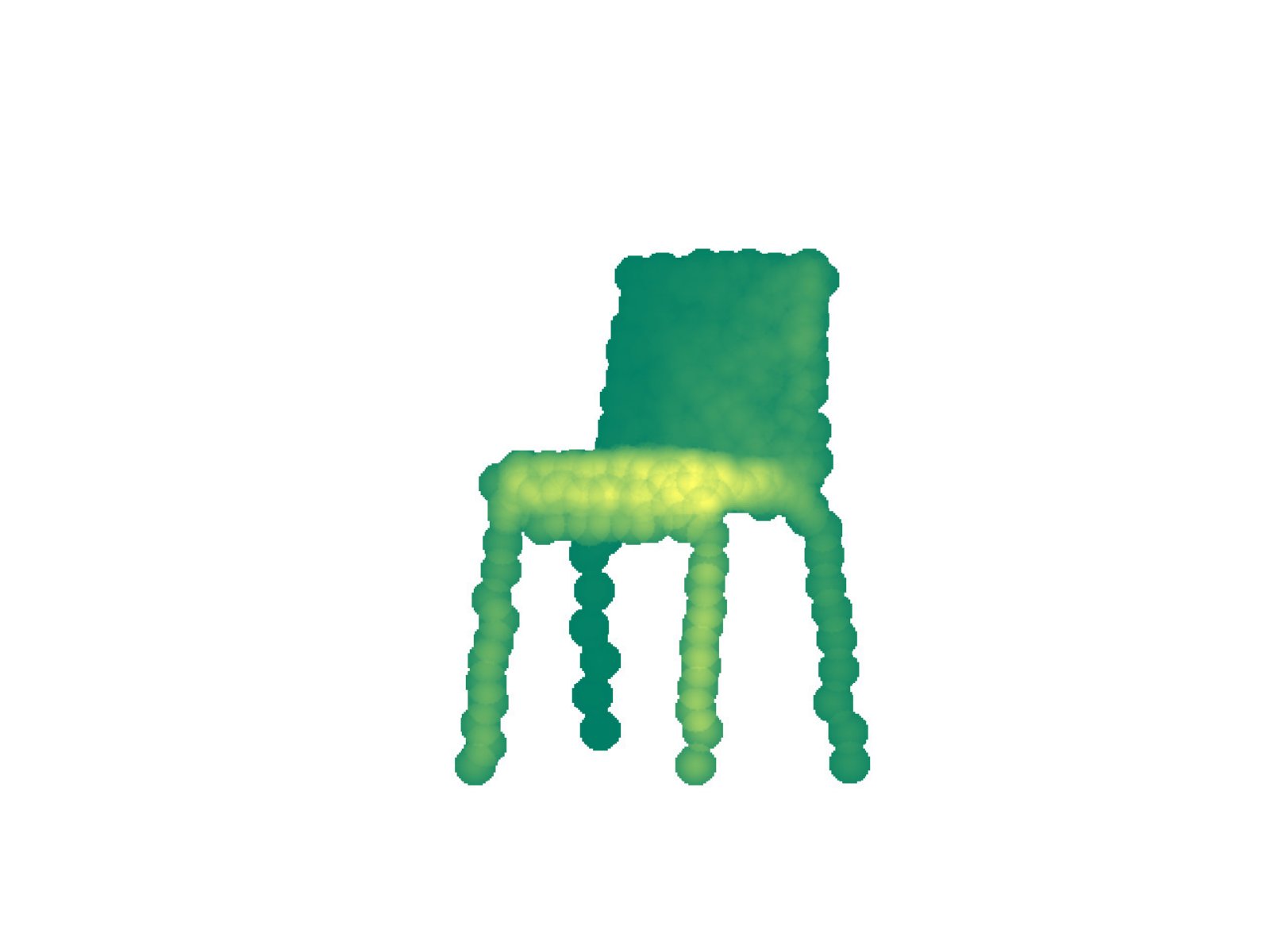} 
\includegraphics[trim=60 0 50 0,width=0.16\textwidth]{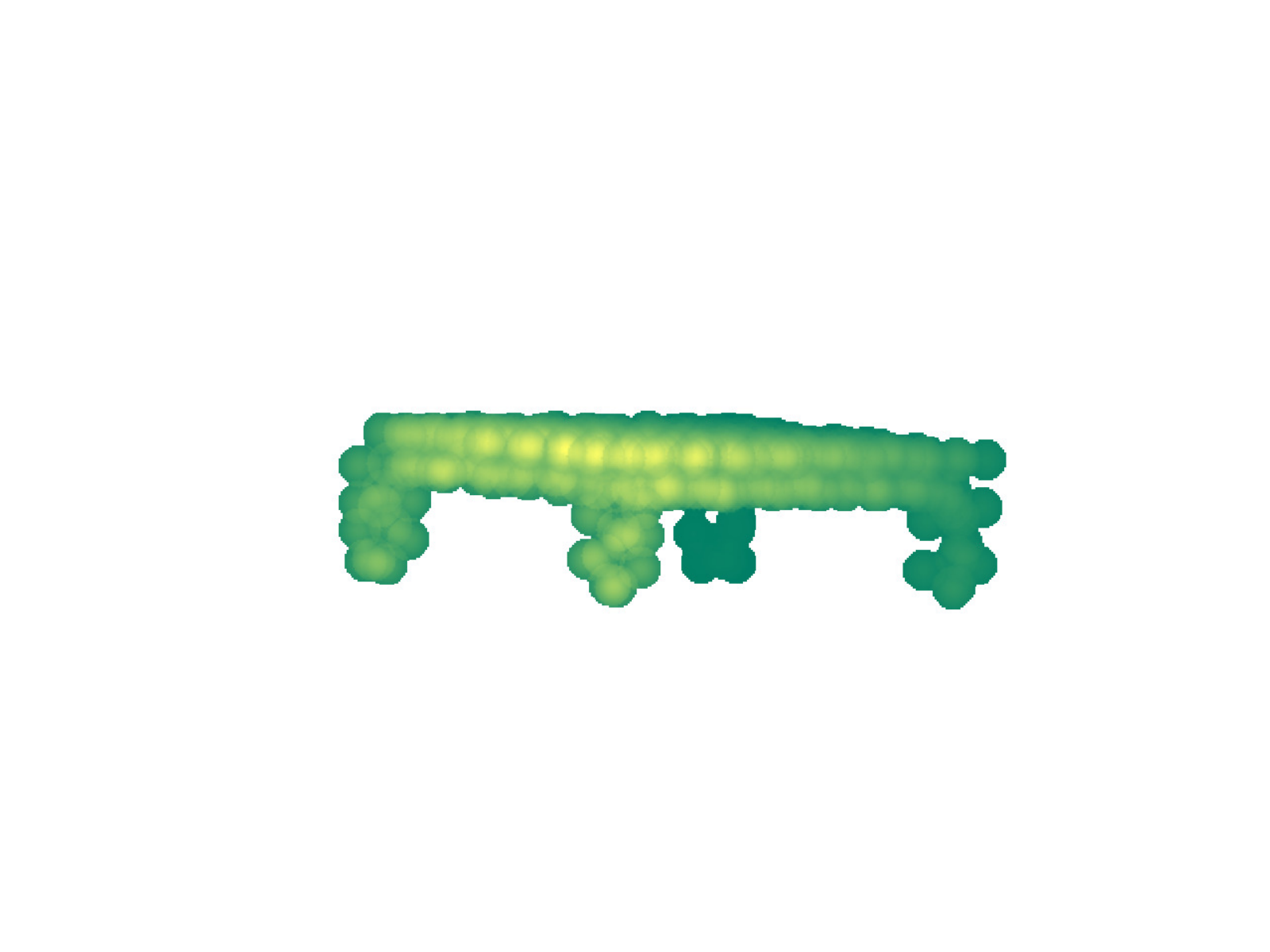}  
\includegraphics[trim=60 0 100 0,width=0.12\textwidth]{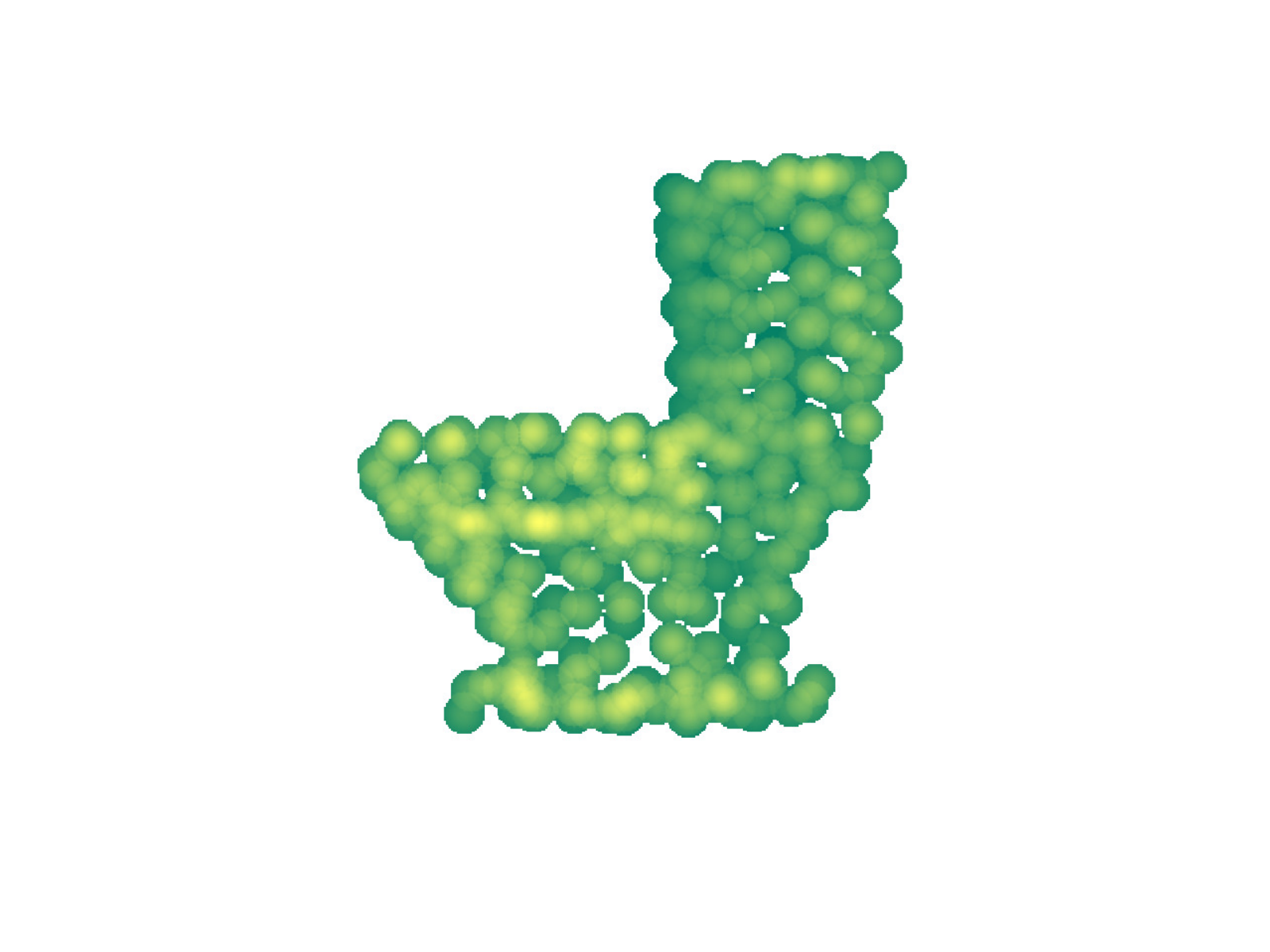}  
\includegraphics[trim=60 0 100 0,width=0.12\textwidth]{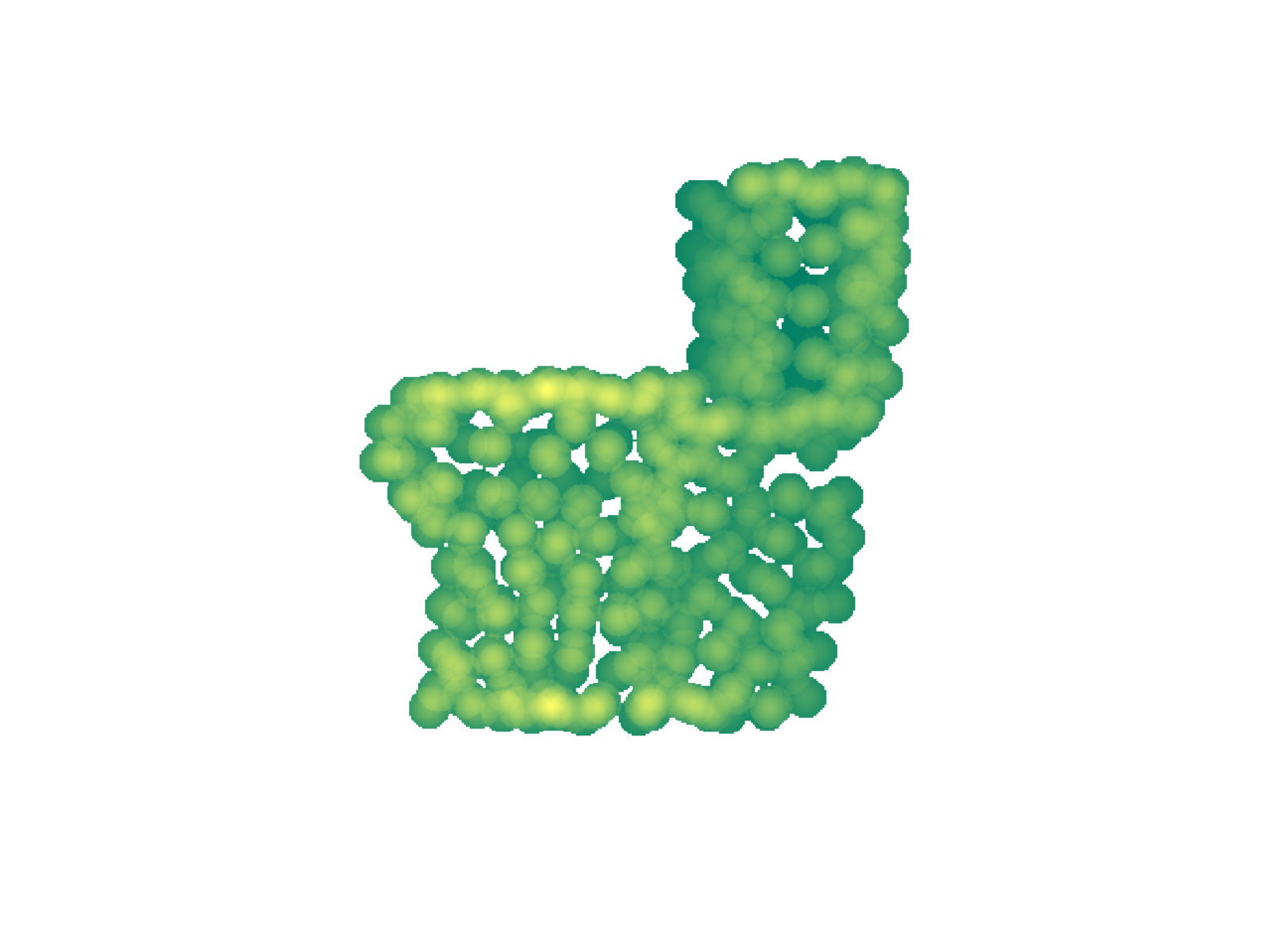} 

\caption{Point cloud models with 300 sampling points in each model }
\label{fig:points}
\end{figure*}

We evaluate the performance of our proposed MNN structure with practical implementations in \eqref{eqn:dis-mnn} and \eqref{eqn:discrete_manifold_convolution_discrete}, i.e., we approximate the unrealizable MNN with GNNs. We carry out a classification problem with ModelNet10 dataset \cite{wu20153d}. The dataset contains 3,991 meshed CAD models from 10 categories for training and 908 models for testing. We sample 300 points from each meshed model uniformly and construct a point cloud. Explicitly, we see each point as a node and the edge weights are calculated with \eqref{eqn:weight}. Our goal is to identify the models for chair from other models as illustrated in Figure \ref{fig:points}. 

The point cloud models are here approximated by the dense graphs. The edge weights are calculated according to \eqref{eqn:weight} with $t_n$ set as 0.3. The Laplacian matrix is calculated with \eqref{eqn:Laplacian-matrix} for each point cloud model. We implement several different architectures to solve the classification problem, including graph filters (GF) and graph neural networks (GNN) with 1 and 2 layers respectively. The single layer architectures contain $F_0=3$ input features which are the coordinates of every point in 3d space. The output features are set as $F_1=64$. The architectures with 2 layers include another layer with $F_2=32$ features. The filter taps in each layer are set as $K=5$. The nonlinearity is a ReLu function. All the architectures are concluded with a linear readout layer mapping the output features to a binary scalar to estimate the classifications.

We train all the architectures with an ADAM optimizer \cite{kingma2014adam} with the learning rate as 0.005 and decaying factor as 0.9,0.999 to minimize the entropy loss. The training model set is divided into batches with 10 models over 40 epochs. We repeat 5 sampling realizations for all the architectures and calculate the average classification error rates as well as the standard deviation in Table \ref{tb:results}. We observe that GNNs perform better than the GFs while architectures with more layers learn more accurate models with more parameters learned.

\begin{table}[h]
\centering
\begin{tabular}{l|c} \hline
Architecture    & error rates   \\ \hline
GNN1Ly	& $8.04 \% \pm 0.88\% $   \\ \hline
GNN2Ly		& $4.30\% \pm 2.64\%$   \\ \hline
GF1Ly		& $13.77\% \pm 6.87\%$   \\ \hline
GF2Ly	& $12.22\% \pm 7.89\%$   \\ \hline
\end{tabular}
\caption{Classification error rates for model `chair' in the test dataset. Average over 5 data realizations. }
\label{tb:results}
\vspace{-3mm}
\end{table}

%% file: conclusion.tex

In this paper, we have studied the limit of a graph sequence as a manifold. We have defined a manifold convolution operation with the heat diffusion controlled by the Laplace-Beltrami operator. We have further constructed a manifold neural network architecture. To realize the MNN in practice, we have carried out discretization in both space and time domains which recovers the convolution and neural networks on graphs. We finally verified the performance of MNN with a model classification problem.

%% file: appendix.tex
 \subsection{Proof of Proposition \ref{prop:convergence}}
 \label{app:convergence}
For ease of presentation, we denote the norm $\|\cdot\|_{L^2(\bbG_n)}$ as $\|\cdot\|$ for short. Considering that the discrete points $\{x_1,x_2,\hdots,x_n\}$ are uniformly sampled from manifold with measure $\mu$, the empirical measure associated with $\text{d}\mu$ is defined as $p_n=\frac{1}{n}\sum_{i=1}^n \delta_{x_i}$, where $\delta_{x_i}$ is the dirac measure supported on $x_i$. Similar to the inner product defined in the $L^2(\ccalM)$ space \eqref{eqn:innerproduct}, the inner product on $L^2(\bbG_n)$ is denoted as
 \begin{equation}
     \langle u, v\rangle_{\bbG_n}=\int u(x)v(x)\text{d}p_n=\frac{1}{n}\sum_{i=1}^n u(x_i)v(x_i).
 \end{equation}
 The norm in $L^2(\bbG_n)$ is therefore $\|u\|^2_{\bbG_n} = \langle u, u \rangle_{\bbG_n}$, with $u,v \in L^2(\ccalM)$. For signals $\bbu,\bbv \in L^2(\bbG_n)$, the inner product is therefore $\langle \bbu,\bbv \rangle_{L^2(\bbG_n)} = \frac{1}{n}\sum_{i=1}^n [\bbu]_i[\bbv]_i$.
 
We introduce the definition of Lipschitz continuous manifold filters in Definition \ref{def:lipschitz} and the assumption on the amplitude of frequency response.
 
 \begin{definition}[Manifold filter with Lipschitz continuity] \label{def:lipschitz}
A manifold filter is $C$-Lispchitz if its frequency response is Lipschitz continuous with constant $C$, i.e,
\begin{equation}
    |\hat{h}(a)-\hat{h}(b)| \leq C |a-b|\text{ for all } a,b\in (0,\infty)\text{.}
\end{equation}
\end{definition}

\begin{assumption}[Non-amplifying filters] \label{ass:filter_function}
A manifold filter is non-amplifying if for all $\lambda\in(0,\infty)$, its frequency response $\hhath$ satisfies $|\hhath(\lambda)|\leq 1$.
\end{assumption}

Note that this assumption is reasonable, because the filter function $\hhath(\lambda)$ can always be normalized.

 We first import the existing results from \cite{belkin2006convergence} which indicates the spectral convergence of the constructed Laplacian operator based on the discretized manifold to the LB operator of the underlying manifold.
 \begin{theorem}[Theorem 2.1 \cite{belkin2006convergence}]
 \label{thm:convergence}
 Let $X=\{x_1, x_2,...x_n\}$ be a set of $n$ sampling points i.i.d. from manifold signal $f$ of $d$-dimensional manifold $\ccalM \subset \reals^N$, sampled by an operator $\bbP_n$ \eqref{eqn:sampling}. Let $\bbG_n$ be a discrete approximation of $\ccalM$ constructed from $X$ with weight values set as \eqref{eqn:weight} with $t_n = n^{-1/(d+2+\alpha)}$ and $\alpha>0$. Let $\bbL_n$ be the graph Laplacian of $\bbG_n$ and $\ccalL$ be the Laplace-Beltrami operator of $\ccalM$. Let $\lambda_{i}^n$ be the $i$-th eigenvalue of $\bbL_n$ and $\bm\phi_{i}^n$ be the corresponding normalized eigenfunction. Let $\lambda_i$ and $\bm\phi_i$ be the corresponding eigenvalue and eigenfunction of $\ccalL$ respectively. Then it holds  that
\begin{equation}
\label{eqn:convergence_spectrum}
    \lim_{n\rightarrow \infty } \lambda_i^n = \lambda_i, \quad \lim_{n\rightarrow \infty} \|\bm\phi^{n}_i(x_j) -  \bm\phi_i(x_j)\|=0, j=1,2 \hdots,n
\end{equation}
where the limits are taken in probability.
 \end{theorem}


With the definitions of neural networks on discretized manifold and manifold, the output difference can be written as 
 \begin{align}
    \nonumber \|\bm\Phi(\bbH,\bbL_n,\bbP_nf)-\bbP_n \bm\Phi&(\bbH,\ccalL, f))\| = \left\| \sum_{q=1}^{F_L}\bbx_{n,L}^q-\sum_{q=1}^{F_L}\bbP_n f_L^q \right\|\\
     & \leq \sum_{q=1}^{F_L} \left\| \bbx_{n,L}^q- \bbP_n f_L^q \right\|.
 \end{align}
 By inserting the definitions, we have 
 \begin{align}
   \nonumber  &\left\| \bbx_{n,l}^p- \bbP_n f_l^p \right\|\\
     &=\left\| \sigma\left(\sum_{q=1}^{F_{l-1}} \bbh_l^{pq}(\bbL_n) \bbx_{n,l-1}^q \right) -\bbP_n \sigma\left(\sum_{q=1}^{F_{l-1}} \bbh_l^{pq}(\ccalL) f_{l-1}^q\right) \right\|
 \end{align}
 with $\bbx_{n,0}=\bbP_n f$ as the input of the first layer. With a normalized point-wise Lipschitz nonlinearity, we have
  \begin{align}
    \| \bbx_{n,l}^p - \bbP_n f_l^p & \| \leq \left\|  \sum_{q=1}^{F_{l-1}} \bbh_l^{pq}(\bbL_n) \bbx_{n,l-1}^q    - \bbP_n \sum_{q=1}^{F_{l-1}} \bbh_l^{pq}(\ccalL)  f_{l-1}^q\right\|\\
    & \leq \sum_{q=1}^{F_{l-1}} \left\|    \bbh_l^{pq}(\bbL_n) \bbx_{n,l-1}^q    - \bbP_n   \bbh_l^{pq}(\ccalL)  f_{l-1}^q\right\|
 \end{align}
 The difference can be further decomposed as
\begin{align}
   \nonumber   \|    \bbh_l^{pq}(\bbL_n) & \bbx_{n,l-1}^q    - \bbP_n   \bbh_l^{pq}(\ccalL)  f_{l-1}^q \| 
   \\ \nonumber&\leq \|
\bbh_l^{pq}(\bbL_n) \bbx_{n,l-1}^q  - \bbh_l^{pq}(\bbL_n) \bbP_n f_{l-1}^q \\ &\qquad +\bbh_l^{pq}(\bbL_n) \bbP_n f_{l-1}^q  - \bbP_n   \bbh_l^{pq}(\ccalL)  f_{l-1}^q
    \|\\\nonumber
   & \leq \left\|
    \bbh_l^{pq}(\bbL_n) \bbx_{n,l-1}^q  - \bbh_l^{pq}(\bbL_n) \bbP_n f_{l-1}^q
    \right\|
  \\ &\qquad +
    \left\|
    \bbh_l^{pq}(\bbL_n) \bbP_n f_{l-1}^q  - \bbP_n   \bbh_l^{pq}(\ccalL)  f_{l-1}^q
    \right\|
\end{align}
The first term can be bounded as $\| \bbx_{n,l-1}^q - \bbP_nf_{l-1}^q\|$ with the initial condition $\|\bbx_{n,0} - \bbP_n f_0\|=0$. The second term can be denoted as $D_{l-1}^n$. With the iteration employed, we can have
\begin{align}
 \nonumber \|\bm\Phi(\bbH,\bbL_n,\bbP_n f) - \bbP_n \bm\Phi(\bbH,\ccalL,f)\| 
 \leq
 \sum_{l=0}^L \prod\limits_{l'=l}^L F_{l'} D_l^n.
 \end{align}
 Therefore, we can focus on the difference term $D_l^n$, we omit the feature and layer index to work on a general form. Considering that $f$ is $\lambda_M$-bandlimited, we can write the convolution operation as follows.
 \begin{align}
    & \nonumber\|\bbh(\bbL_n)\bbP_n f - \bbP_n\bbh(\ccalL) f\|
   \\& \leq \left\| \sum_{i=1}^M \hat{h}(\lambda_i^n) \langle \bbP_nf,\bm\phi_i^n \rangle_{\bbG_n}\bm\phi_i^n - \sum_{i=1}^M \hat{h}(\lambda_i)\langle f,\bm\phi_i\rangle_{\ccalM} \bbP_n \bm\phi_i  \right\|
     \\ 
     &\nonumber \leq  \left\| \sum_{i=1}^M \hat{h}(\lambda_i^n) \langle \bbP_nf,\bm\phi_i^n \rangle_{\bbG_n}\bm\phi_i^n - \sum_{i=1}^M \hat{h}(\lambda_i) \langle \bbP_nf,\bm\phi_i^n \rangle_{\bbG_n}\bm\phi_i^n\right\| \\
     & \quad +\left\| \sum_{i=1}^M \hat{h}(\lambda_i) \langle \bbP_n f,\bm\phi_i^n \rangle_{\bbG_n} \bm\phi_i^n - \sum_{i=1}^M \hat{h}(\lambda_i) \langle f,\bm\phi_i \rangle_{\ccalM} \bbP_n \bm\phi_i \right\|,\label{eqn:conv-1}
 \end{align}
 where $M=\#\{\lambda_i\leq \lambda_M\}_i$ counts the number of eigenvalues within the bandwidth. 
The first term in \eqref{eqn:conv-1} can be bounded by the $C$-Lipschitz continuity of the frequency response function. From the convergence in probability stated in \eqref{eqn:convergence_spectrum}, we can claim that for each eigenvalue $\lambda_i \leq \lambda_M$, for all $\epsilon_i>0$ and all $\delta_i>0$, there exists some $N_i$ such that for all $n>N_i$, we have
\begin{gather}
 \label{eqn:eigenvalue}   \mathbb{P}(|\lambda_i^n-\lambda_i|\leq \epsilon_i)\geq 1-\delta_i,
 \end{gather}
Letting $\epsilon_i < \epsilon$ with $\epsilon > 0$, with probability at least $\prod_{i=1}^M(1-\delta_i) := 1-\delta$, the first term is bounded as 
\begin{align}
   &\left\| \sum_{i=1}^M (\hat{h}(\lambda_i^n) - \hat{h}(\lambda_i)) \langle \bbP_n f,\bm\phi_i^n \rangle_{\bbG_n} \bm\phi_i^n  \right\|
   \\
   &\qquad \leq \sum_{i=1}^M |\hat{h}(\lambda_i^n)-\hat{h}(\lambda_i)| |\langle \bbP_n f,\bm\phi_i^n \rangle_{\bbG_n}| \|\bm\phi_i^n\|\\
   &\qquad \leq \sum_{i=1}^M C |\lambda_i^n-\lambda_i| \|\bbP_n f\| \|\bm\phi_i^n \|^2\leq MC\epsilon,
\end{align} 
for all $n>\max_i N_i := N$.
 
The second term in \eqref{eqn:conv-1} can be bounded combined with the convergence of eigenfunctions in \eqref{eqn:eigenfunction} as
\begin{align}
   &\nonumber \left\| \sum_{i=1}^M \hat{h}(\lambda_i) \langle \bbP_nf,\bm\phi_i^n \rangle_{\bbG_n}\bm\phi_i^n - \sum_{i=1}^M \hat{h}(\lambda_i) \langle f,\bm\phi_i \rangle_{\ccalM} \bbP_n \bm\phi_i \right\|\\
   &\leq \nonumber \left\|  \sum_{i=1}^M \hat{h}(\lambda_i) \left(\langle \bbP_n f,\bm\phi_i^n\rangle_{\bbG_n}\bm\phi_i^n  - \langle \bbP_nf,\bm\phi_i^n \rangle_{\bbG_n} \bbP_n\bm\phi_i\right)\right\|\\
   &\label{eqn:term1}+ \left\| \sum_{i=1}^M  \hat{h}(\lambda_i) \left(\langle \bbP_n f,\bm\phi_i^n\rangle_{\bbG_n} \bbP_n\bm\phi_i -\langle f,\bm\phi_i\rangle_\ccalM \bbP_n\bm\phi_i \right) \right\|
\end{align}
From the convergence in probability stated in \eqref{eqn:convergence_spectrum}, we can claim that for some fixed eigenfunction $\bm\phi_i$,  for all $\epsilon_i>0$ and all $\delta_i>0$, there exists some $N_i$ such that for all $n>N_i$, we have
\begin{gather}
 \label{eqn:eigenfunction}    \mathbb{P}(|\bm\phi_i^n(x_j) - \bm\phi_i(x_j)|\leq \epsilon_i)\geq 1-\delta_i,\quad \forall x_j\in X .
 \end{gather}
 Therefore, letting $\epsilon_i < \epsilon$ with $\epsilon > 0$, with probability at least $\prod_{i=1}^M(1-\delta_i) := 1-\delta$, for all $n> \max_i N_i := N$, the first term in \eqref{eqn:term1} can be bounded as
\begin{align}
&\nonumber \left\|  \sum_{i=1}^M \hat{h}(\lambda_i) \left(\langle \bbP_n f,\bm\phi_i^n\rangle_{\bbG_n}\bm\phi_i^n  - \langle \bbP_nf,\bm\phi_i^n \rangle_{\ccalM} \bbP_n\bm\phi_i\right)\right\|\\
&\qquad \qquad \qquad \leq \sum_{i=1}^M \|\bbP_n f\|\|\bm\phi_i^n - \bbP_n\bm\phi_i\|\leq M\epsilon,
\end{align}
considering the boundedness of frequency response function. The last equation comes from the definition of norm in $L^2(\bbG_n)$.
The second term in \eqref{eqn:term1} can be written as
\begin{align}
  \nonumber &\left\| \sum_{i=1}^M  \hat{h}(\lambda_i^n) \left(\langle \bbP_n f,\bm\phi_i^n\rangle_{\bbG_n} \bbP_n\bm\phi_i -\langle f,\bm\phi_i\rangle_\ccalM \bbP_n\bm\phi_i \right) \right\| \\
   &\leq \sum_{i=1}^M|\hat{h}(\lambda_i^n)| \left|\langle \bbP_n f,\bm\phi_i^n\rangle_{\bbG_n}  -\langle f,\bm\phi_i\rangle_\ccalM\right|\|\bbP_n\bm\phi_i\|.
\end{align}
Because $\{x_1, x_2,\cdots,x_n\}$ is a set of uniform sampled points from $\ccalM$, based on Theorem 19 in \cite{von2008consistency} we can claim that 
\begin{equation}
    \lim_{n\to \infty} \mathbb{P}\left(\left|\langle \bbP_n f,\bm\phi_i^n\rangle_{\bbG_n}  -\langle f,\bm\phi_i\rangle_\ccalM\right|\leq\epsilon \right)\geq 1-\delta,
\end{equation}
for all $\epsilon>0$ and $\delta>0$.Taking into consider the boundedness of frequency response $|\hat{h}(\lambda)|\leq 1$ and the bounded energy $\|\bbP_n\bm\phi_i\|$. Therefore, we have for all $\epsilon>0$ and $\delta>0$,
\begin{align}
 \nonumber   & \lim_{n\to \infty}\mathbb{P}\left(\left\| \sum_{i=1}^M  \hat{h}(\lambda_i^n) \left(\langle \bbP_n f,\bm\phi_i^n\rangle_{\bbG_n}  -\langle f,\bm\phi_i\rangle_\ccalM \right)\bbP_n\bm\phi_i  \right\|\leq M \epsilon\right)\\
 &\qquad \qquad \qquad\qquad \qquad \qquad \qquad  \qquad \geq 1-\delta.
\end{align}

Combining all these results, we can claim that for all $\epsilon'>0$ and $\delta>0$, there exists some $N$, such that for all $n>N$ we have
\begin{equation}
    \mathbb{P}(\|\bbh(\bbL_n)\bbP_n f - \bbP_n\bbh(\ccalL) f\|\leq \epsilon')\geq 1-\delta.
\end{equation}

With $\lim\limits_{n\rightarrow \infty}D_l^n=0$ in high probability, this concludes the proof. 